\documentclass[10pt, conference, letterpaper]{IEEEtran}
\makeatletter
\def\ps@headings{%
\def\@oddhead{\mbox{}\scriptsize\rightmark \hfil \thepage}%
\def\@evenhead{\scriptsize\thepage \hfil \leftmark\mbox{}}%
\def\@oddfoot{}%
\def\@evenfoot{}}
\makeatother 

\ifCLASSINFOpdf

\else

\fi

\usepackage[cmex10]{amsmath}
\makeatletter
\def\ps@headings{%
\def\@oddhead{\mbox{}\scriptsize\rightmark \hfil \thepage}%
\def\@evenhead{\scriptsize\thepage \hfil \leftmark\mbox{}}%
\def\@oddfoot{}%
\def\@evenfoot{}}
\makeatother 

\IEEEoverridecommandlockouts
\usepackage{bbm}
\usepackage{amsfonts}
\usepackage[dvips]{graphicx}
\usepackage{times}
\usepackage{cite}
\usepackage{amsmath}
\usepackage{amsthm}
\usepackage{array}
\usepackage{amssymb}

\usepackage{stfloats}
\usepackage{slashbox}
\usepackage{graphicx}
\usepackage{footnote}
\usepackage{booktabs}
\usepackage{array}
\usepackage{algorithm}
\usepackage{subeqnarray}
\usepackage{cases}
\usepackage{threeparttable}
\usepackage{color}
\usepackage{hyperref}
\usepackage{epstopdf}
\usepackage{algpseudocode}
\usepackage{bm}
\usepackage{multirow}
\usepackage[labelformat=simple]{subcaption}
\usepackage{adjustbox}
\usepackage{amsmath}
\usepackage{amssymb}

\newtheorem{theorem}{Theorem}
\newtheorem{remark}{Remark}
\newtheorem{proposition}{Proposition}

\newtheorem{definition}{Definition}

\begin{document}
%
\title{Caching Incentive Design in Wireless D2D Networks: A Stackelberg Game Approach}
\author{\IEEEauthorblockN{Zhuoqun Chen, Yangyang Liu, Bo Zhou and Meixia Tao}
\IEEEauthorblockA{Dept. of Electronic Engineering, Shanghai Jiao Tong University, Shanghai, China\\
 Email: \{billchen, bingpumpkn, b.zhou, mxtao\}@sjtu.edu.cn}
\thanks{This work is supported by the National Natural Science Foundation of China under grants 61571299 and  61329101.}
 }


\maketitle

\begin{abstract}
Caching in wireless device-to-device (D2D) networks can be utilized to offload data traffic during peak times. However, the design of incentive mechanisms is challenging due to the heterogeneous preference and selfish nature of user terminals (UTs). In this paper, we propose an incentive mechanism in which the base station (BS) rewards those UTs that share contents with others using D2D communication. We study the cost minimization problem for the BS and the utility maximization problem for each UT. In particular, the BS determines the rewarding policy to minimize his total cost, while each UT aims to maximize his utility by choosing his caching policy. We formulate the conflict among UTs and the tension between the BS and the UTs as a Stackelberg game. We show the existence of the equilibrium and propose an iterative gradient algorithm (IGA) to obtain the Stackelberg Equilibrium. Extensive simulations are carried out to evaluate the performance of the proposed caching scheme and comparisons are drawn with several baseline caching schemes with no incentives. Numerical results show that the caching scheme under our incentive mechanism outperforms other schemes in terms of the BS serving cost and the utilities of the UTs.
\end{abstract}

\section{Introduction}

The mobile data traffic is growing dramatically in recent years\cite{Cisco}. 
To deal with this exponential growth of traffic, content caching is introduced
because of the relatively low storage cost of wireless devices\cite{femto}.
Helper stations or user terminals (UTs) decide what contents to cache and proactively fetch them via backhaul during off-peak times, and
transmit the content to requesters during peak times. In this way, caching offloads the network traffic during peak times and reduces UTs' average delay cost.


Device-to-device (D2D) communication \cite{song2015wireless} can be utilized to enhance the benefits of caching.
D2D communication enables a UT to communicate directly with other UTs in his vicinity. When a UT's content request cannot be satisfied by his local cache, any of his neighbors who cache that content can become the content server and transmit the content using D2D communication. This technique improves spectrum utilization, increases network throughput, and reduces average access delay for UTs \cite{molisch2014caching}.
However, in reality, UTs are selfish and only care about their own preferences.
Each UT intends to cache his favorite contents and also hopes that his neighbors could cache
his other favorite contents as many as possible.

Consider a simple example of D2D caching network with 4 UTs and 10 contents $\{1,\cdots,10\}$. Let the size
of each content and the size of each UT's cache be 1. Assume that all UTs can communicate with each other.
Also assume that the content preference ranking of UT 1 and UT 2 are $(1,2,3,\cdots,10)$, $(10,9,8,\cdots,1)$,
respectively.
Naturally, UT 1 would cache content $1$ and wishes his three neighbors to cache contents $2$, $3$ and $4$,
whereas UT 2 would cache content $10$ and hopes that his neighbors would cache contents $9$, $8$ and $7$.
The difference in preference causes conflict of caching interest among UTs, which cannot be reconciled without intervention.

On the other hand, the base station (BS) aims to minimize the traffic load of serving the UTs, thus reducing the backhaul load and the transmission cost. This goal is equivalent to maximizing the chances of D2D communications among UTs.
The nature of selfishness of UTs however, becomes the major obstacle for the BS to achieve this goal.
Consider a D2D network with selfish UTs. Each UT cares merely about his own preference and only caches the contents he likes most.
This may cause duplicate caching and underutilization of the storage space for all UTs.
Consequently, the BS is overloaded by UTs' requests and the UTs suffer from larger delays.
Therefore, it is essential for the BS to introduce incentive mechanisms into the network to motivate UTs to cache in a way to promote D2D communication.

The interaction between the BS and the UTs, two rational entities with conflicting objectives, is often characterized using game theory.
Game theory is the study of mathematical models of conflict and cooperation between intelligent rational decision-makers \cite{myerson1991game}.
Specifically, the hierarchical relationship between the BS and the UTs best suits the Stackelberg game model.
In this game, the BS is the leader and decides the incentive mechanism, i.e. the rewarding policy.
The UTs are the followers and respond to the rewarding policy with their caching decisions of what contents to cache.

Various works have modeled the interaction between the BS and end-users by Stackelberg game models \cite{poularakis2014framework},\cite{wu2012game}.
However, limited work has addressed the circumstance where the utility of end-users are coupled.
In other words, conflicts also exist among the followers\cite{alotaibi2015game}.
In this paper, we model the interactions among the UTs as a UT sub-game and analyze the existence of the
Nash Equilibrium of the sub-game.

The contributions of this paper are as follows:
\begin{itemize}
\item We propose an incentive mechanism for UTs to cache in order to promote chances of D2D communication.
We model the interaction between the BS and the UTs as a Stackelberg game, and the conflicts among UTs as a sub-game of
the Stackelberg game.
\item We analyze our incentive mechanism in a special case of two UTs and two contents to obtain some insights of the effect of the rewarding policy of the BS on the caching decisions of the UTs.
\item We propose an iterative gradient algorithm (IGA) to obtain the Stackelberg Equilibrium (SE). The optimal rewarding
policy of the BS and the corresponding caching policies of all UTs are achieved at the SE.
\item We verify the effectiveness of our proposed caching scheme using the IGA algorithm by extensive simulations.
Numerical results show that our scheme outperforms other caching schemes with no incentives in terms of BS serving cost and the utilities of UTs .
\end{itemize}


The rest of the paper is organized as follows. Section II presents the system model. In Section III, we give the formulation of Stackelberg game. Section IV presents a special case with two UTs and two contents. We present the Iterative Gradient Algorithm in Section V. Section VI provides performance evaluation results, and we conclude the paper in Section VII. The important notations used in this paper are summarized in Table I.

\section{System Model}
\subsection{Network Model}
We consider a single cell wireless D2D network consisting of one BS and $N$ UTs.
Let $\mathcal{N}=\left\{1,2,\cdots,N\right\}$ denote the set of UTs.
Let $\mathcal{N}_{i}\subseteq \mathcal{N}\setminus\{i\}$ denote the set of neighbors of UT $i$. Each UT can communicate and share contents directly with all his neighbors via D2D links.
Denote $\mathcal{N}_{i}^+=\mathcal{N}_{i}\cup\{i\}$.
We denote $d_{i,j}$ as the delay cost between UT $i$ and UT $j$ using D2D communication. With abuse of notation, we set $d_{i,i}=0$ for all $i\in\mathcal{N}$.

\subsection{Cache Model}
Let $\mathcal{M}=\left\{1,2,\cdots,M\right\}$ denote the set of the contents in the network.
Each UT has a limited cache size, denoted as $c_{i}$.
We represent the cache states of UTs by using an $N\times M$ matrix $X$,
where $x_{i}^m$ indicates the proportion (between 0 and 1) of content $m$ that UT $i$ caches.
We assume that all the contents are encoded by rateless MDS coding (e.g., using Raptor codes \cite{shokrollahi2006raptor}).
With MDS coding, a content can be retrieved given that the receiver has received more than a certain
number of bits of the content in any order. Therefore, a request of a UT can be satisfied via multiple D2D
communications between various other UTs, as long as the total amount of the received bits exceeds a proportion-threshold.
Without loss of generality, we assume that this proportion-threshold is 1 for any content to be successfully retrieved.

We consider a heterogeneous request model, where the preferences of the UTs are different and
hence the popularity of one content varies from one UT to another. Let $p_i^m$ denote the probability of UT $i$ requesting content $m$, where $i\in \mathcal{N}$ and $m\in \mathcal{M}$. For each UT $i$, we require $\sum_{m=1}^{M}p_i^m =1$.
We also assume that UTs and BS know the preferences of every UT perfectly, i.e. $(p_i^m)_{i\in \mathcal{N},m\in \mathcal{M}}$ is a common knowledge within the network. Note that, the preference of each UT evolves at a timescale much slower than the timescale of content requesting, and it can be learned accurately by monitoring his activity \cite{femto}.
\vspace{-0.05cm}
\subsection{Service Model}
Each UT can be a requester and a content provider.
For instance, we consider UT $i$ requesting for content $m$.
He first checks his own cache. If there exists a whole copy of content $m$ in
his cache, the request is satisfied with no delay. Otherwise, he
asks for his neighbors who cache (a portion of) content $m$ and uses D2D service to retrieve the content.

To serve the request of UT $i$, we sort the D2D link delay costs $(d_{i,j})_{j\in\mathcal{N}_i^+}$ in an increasing order.
Let $(i)_j$ denote the  index of the UT with the $i$th lowest delay cost to UT $j$. Note that $(1)_j=j$.
First, UT $i$ checks his own cache and receives the service if certain portions of the content are cached.
If this service is not enough to satisfy the request, the UT with the second lowest delay cost serves UT $i$, i.e., the nearest UT of UT $i$.
This process continues until UT $i$ has obtained as many bits as the proportion-threshold for content $m$
(including the portion in his own cache).

If all other UTs' D2D service is still not enough to recover the
requested content for UT $i$, the rest part of the content is finally served by the BS. Let $d_{i,0}$ denote the
delay cost between UT $i$ and the BS. We assume that $d_{i,0} \gg d_{i,j}$, for any UT $j$, indicating that the delay
cost of service from the BS is much larger than that of any D2D service.

%


\begin{table}[tbp]
\caption{LIST OF IMPORTANT NOTATIONS}\label{tablenotation}
\begin{tabular}{|c|c|}
\hline
$\mathcal{M}=\{1,2,\cdots,M\}$ & set of contents\\
\hline
$\mathcal{N}=\{1,2,\cdots,N\}$ & set of UTs\\
\hline
$\mathcal{N}_{i}\subseteq \mathcal{N}\setminus\{i\}$ & set of neighbors of UT $i$\\
\hline
$c_i$ &  cache size of UT $i$\\
\hline
$s_m$ &  size of content $m$\\
\hline
$\mathbf{x}_i=(x_{i}^m)_{m\in\mathcal{M}}\in\left[0,1\right]^{1\times M}$ & $x_{i}^m$ is the proportion of \\$\mathbf{X}=(\mathbf{x}_i)_{i\in\mathcal{N}}\in\left[0,1\right]^{N\times M}$
& content $m$ that is cached by UT $i$.\\
\hline
$d_{i,j}$ & delay cost of transmitting one bit\\&  between UT $i$ and UT $j$\\
\hline
$d_{i,0}$ & delay cost of transmitting one bit\\&  from BS to UT $i$\\
\hline
$(i)_j$& UT index with the $i$th lowest\\& delay cost to UT $j$, $(1)_j=j$\\
\hline
$[i]_j$ & ranking of $d_{i,j}$ in $(d_{l,j})_{l\in\mathcal{N}_j^+}$,\\ &$[j]_j=1$ \\
\hline
$p_i^m$ &  probability that UT $i$\\&  requests content $m$\\
\hline
$r$&  unit reward paid from BS to UT\\ & for D2D service\\
\hline
$w_s$ & unit serving cost of BS\\& for serving UTs' requests\\
\hline
$w_d$ & weight of delay cost\\
\hline
\end{tabular}
\centering

\end{table}

\section{Stackelberg Game Formulation}
\subsection{Incentive Mechanism}
We design an incentive mechanism under which the BS rewards the UTs based on the amount of content they serve their neighbor UTs. The mechanism can enhance the chances of D2D transmission between UTs and hence release the burden on BS. The UTs can benefit from receiving rewards by serving other UTs and enjoying smaller delay with more D2D service. The BS can benefit from reducing the workload and thus lowering the operational cost.
Note that the reward could be of any form, such as monetary value or virtual credits, and is paid to the UTs under specific protocols, all of which are out of the scope of this paper.

Stackelberg game is an extension of non-cooperative game with a bi-level hierarchy. Stackelberg game models a game between two groups of players, namely leaders and followers. The leaders have the privilege of acting first while the followers act according to the leaders' actions.

We formulate our problem into a single-leader multi-follower Stackelberg game. The BS acts as the leader and the UTs are the followers. The BS first announces the unit reward $r$ of caching for the purpose of D2D. The UTs then determine their caching strategies to maximize their utilities based on the announced reward. The Stackelberg game consists of two sub-problems: the UT sub-game and the BS optimization.

\subsection{UT Sub-game}
In the Stackelberg formulation, each UT finds his optimal caching policy based on the unit reward announced by the BS, as well as the caching policy of all other UTs.
For the UT sub-game, every UT intends to maximize his utility $U_i$ which is the difference of reward received from BS
and his total delay cost. The reward a UT receives from the BS is proportional to the total amount of content he serves his neighbors. The caching placement problem of UT $i$ is given by
\begin{align} \max_{\mathbf{x}_i} & \ U_i = \sum_{m\in\mathcal{M}}\sum_{j\in\mathcal{N}_i} p_j^mrs_{m}F_{i,j}^m - w_d\sum_{m\in\mathcal{M}}p_i^ms_m\bar{D}_{i}^{m} \\
s.t.~&\sum_{m\in\mathcal{M}}x_i^ms_m\leq c_i\\
~&\mathbf{x}_i \in [0,1]^{1\times M}
\end{align}
where $F_{i,j}^m = \min\left\{x_i^m,\max\left\{0,1-\sum_{k=1}^{\left[i\right]_j - 1}x_{\left(k\right)_j}^m\right\}\right\}$
is the portion of content $m$ that UT $i$ serves UT $j$ via D2D communication, and $w_d$ is the weight of the delay cost.
$\bar{D}_{i}^{m}$ is the average delay of UT $i$ requesting content $m$. We adopt the notations in \cite{femto} to define the UT's delay cost. The average delay of UT $i$ requesting content $m$ is given by
$$ \bar{D}_{i}^{m}=\left\{
\begin{aligned}
\bar{D}_{i}^{m,1} & \quad  & \text{if}\ x_{\left(1\right)_i}^m \geq 1\\
\vdots & \quad  & \vdots\\
\bar{D}_{i}^{m,j} & \quad  & \text{if}\ \sum_{k=1}^{j-1}x_{\left(k\right)_i}^m < 1,\ \sum_{k=1}^{j}x_{\left(k\right)_i}^m \geq 1\\
\vdots & \quad  & \vdots\\
\bar{D}_{i}^{m, \left| \mathcal{N}_i \right| +1} & \quad  & \text{if}\ \sum_{k=1}^{\mathcal{N}_i}x_{\left(k\right)_i}^m < 1
\end{aligned}
\right.
$$
where $\bar{D}_i^{m,j}$ denotes the average delay cost per bit for UT $i$ to download content $m$ from his best $j$ neighbors, and is given by
\begin{equation}
\bar{D}_{i}^{m,j} = \sum_{k=1}^{j-1}x_{\left(k\right)_i}^md_{(k)_i,i} + \left(1-\sum_{k=1}^{j-1}x_{\left(k\right)_i}^m\right)d_{(j)_i,i}.
\end{equation}
\\
Constraint (2) indicates that the amount of content in each UT's cache cannot exceed his cache size. Constraint (3) requires that a UT cannot cache more than the size of a content.

\subsection{BS Optimization}

In our Stackelberg game formulation, the BS makes the first move by determining the unit reward of caching for D2D communications to minimize his total cost $C$. The total cost $C$ is made up of two parts: the reward cost and the serving cost. The reward cost is the total amount of reward given to the UTs and the serving cost is the cost of serving all remaining requests of the UTs. We assume that the reward that UT $i$ receives for serving UT $j$ is proportional to the amount of content UT $i$ actually serves. Hence, the BS's total cost is given by
\begin{align}
C = \sum_{m\in\mathcal{M}}\sum_{i\in\mathcal{N}}&p_i^ms_m\Bigg[\sum_{j\in\mathcal{N}_i}rF_{j,i}^m \nonumber\\&+ w_{s}d_{i,0}\max\Big\{0,1-\sum_{j\in\mathcal{N}_i^+}x_{j}^{m}\Big\}\Bigg],
\end{align}
where $w_{s}$ is the unit serving cost for BS. The serving cost of BS to serve UT $i$ is proportional to the delay cost between UT $i$ and the BS.

The BS tries to find the optimal unit reward $r^{*}$ that minimizes the cost $C$. Therefore, the optimal rewarding policy for BS is
\begin{equation}
r^* := \arg \min_{r \geq 0} C
\end{equation}

\subsection{Nash Equilibrium of UT Subgame}
We now show the existence of subgame perfect Nash Equilibrium for the UT sub-game between $N$ followers.
\begin{proposition}
There exists at least one Nash Equilibrium for the UT sub-game.
\end{proposition}

\begin{proof}
The players' strategy space is a closed bounded convex set. Since $\bar{D}_i^m$ is a convex function of $\mathbf{X}$ \cite{femto}, the utility function of UT $i$ is concave in $\mathbf{x}_i$. The utility function $U_i$ is also continuous in $\mathbf{X}$. Therefore, the UT sub-game is a concave game \cite{rosen1965existence}. By Schauder fixed-point theorem \cite{goebel1990topics}, the existence of Nash Equilibrium in UT sub-game is proved.
\end{proof}

\subsection{Stackelberg Equilibrium}
The purpose of the proposed game is to reach the Stackelberg Equilibrium (SE), from which neither the leader (BS) nor the followers (UTs) have any incentive to deviate.
The SE for the proposed game is defined as follows.
\begin{definition}
Let $r^*$ be a solution to the BS optimization and let $\mathbf{x}_i^* = \mathbf{x}_i(r^*)$ be a solution to the UT sub-game of the UT $i$ given the BS reward $r^*$. Then, $(r^*,\mathbf{X}^*)$ is SE for the proposed Stackelberg game if for any $(r,\mathbf{X})$ in the feasible region, the following conditions are satisfied:
\begin{equation}
C(r^*,\mathbf{X}^*)\leq C(r,\mathbf{X}^*),
\end{equation}
\begin{equation}
U_i(\mathbf{x}_i^*,\mathbf{x}_{-i}^*,r^*)\geq U_i(\mathbf{x}_i,\mathbf{x}_{-i}^*,r^*), \forall i \in \mathcal{N}
\end{equation}
\end{definition}


\section{Special Case with Two UTs and Two Contents}
We now consider a special case of two UTs ($N = 2$) and two contents ($M=2$) to gain some insights into the impact of reward value $r$ towards UTs' caching decisions in our model.  We assume that the two contents have the same size, i.e., $s_1=s_2=1$, and the cache size for each UT is 1, i.e., $c_1=c_2=1$. Here we only investigate the caching policies of UT 1 where UT 2 has the same properties as UT 1.

First, we rewrite the utility function of UT 1 (UT 2 has the same format),
\begin{align}
& \max \quad\sum_{m=1}^{2}\big[p_2^mr \min\{x_1^m,1-x_2^m\}- \label{eqn:obj}\\
& w_dp_1^m\max\{\left(1-x_1^m\right)d_{1,2},\left(1-x_1^m-x_2^m\right)d_{1,0}+x_2^md_{1,2}\}\big]\nonumber\\
& s.t. \quad x_1^1 + x_1^2 \leq 1
\end{align}


Then, we have the following properties.

\begin{theorem}
For the case of $N=2$, $M=2$, $s_1=s_2=1$ and $c_1=c_2=1$, the optimal caching policy of UT 1 is as follows.
\begin{itemize}
  \item
  If $0 \leq r<\left|\frac{p_1^1-p_1^2}{p_2^1-p_2^2}\right|w_{d}d_{1,0}$, the optimal caching policy of UT 1 is to completely cache content $m_{1}^{*}$, which is given by
  \begin{equation}
    m_{1}^{*} = \arg\max_{k=1,2}p_{1}^{k}.
  \end{equation}
  \item
  If $r \geq \left|\frac{p_1^1-p_1^2}{p_2^1-p_2^2}\right|w_{d}d_{1,0}$, the optimal caching policy of UT 1 is to completely cache content $m_{1}^{*}$, which is given by
  \begin{equation}
    m_{1}^{*} = \arg\max_{k=1,2}p_{2}^{k}.
  \end{equation}
\end{itemize}

In both cases, UT 2 caches the content which is not cached by UT 1.
\end{theorem}

\begin{proof}
By observing the objective function in \eqref{eqn:obj}, we only need to consider have the following two cases.

In the first case, i.e., $x_1^1+x_2^1 \leq 1$ and $x_1^2+x_2^2 \leq 1$, the optimization problem for UT 1 in \eqref{eqn:obj} can be transformed into:
\begin{align}
\max  & \quad\sum_{m=1}^{2}(p_2^mr+w_dp_1^md_{1,0})x_1^m + (d_{1,0}-d_{1,2})w_dp_1^mx_2^m \nonumber\\
& - w_dp_1^md_{1,0} \\
s.t. & \ 0 \leq x_1^1+x_1^2 \leq 1
\end{align}

We can see that, when $r=0$, then the optimal solution for the above problem is $(x_1^1,x_1^2)=(1,0)$ if $p_1^1-p_1^2\leq 0$, and $(x_1^1,x_1^2)=(0,1)$ otherwise. When $r\neq 0$, the optimal solution is $(x_1^1,x_1^2)=(1,0)$ if $(p_{2}^{1}-p_{2}^{2})r+w_{d}d_{1,0}(p_{1}^{1}-p_{1}^{2})>0$ and $(p_{2}^{1}-p_{2}^{2}) \geq 0$, and $(x_1^1,x_1^2)=(0,1)$ otherwise. That is, when $r-\left|\frac{p_1^1-p_1^2}{p_2^1-p_2^2}\right|w_{d}d_{10}>0$, the UT 1 would cache the content according to UT 2's preference. Moreover, since $x_1^1+x_2^1\leq 1$ and $x_1^2 + x_2^2\leq 1$, we know that UT 2 would cache the content which is different from UT 1.

In the second case, i.e., $x_1^1+x_2^1\leq 1$ and $x_1^2+x_2^2\geq 1$, the optimization problem in \eqref{eqn:obj} can be transformed into:
\begin{align}
\max & \quad w_{d}p_{1}^{1}d_{1,2}x_{1}^{1}-p_{2}^{1}rx_{1}^{1}-p_{2}^{1}rx_{2}^{1}-w_{d}p_{1}^{1}d_{1,2}+p_{2}^{1}r+\nonumber\\
& (p_{2}^{2}r+w_{d}p_{1}^{2}d_{1,0})x_{1}^{2}+(d_{1,0}-d_{1,2})w_{d}p_{1}^{1}x_{2}^{2}-w_{d}p_{1}^{2}\nonumber \\
s.t. & \ 0 \leq x_1^1+x_1^2 \leq 1
\end{align}

This objective function is a linear combination of $x_1^1$ and $x_1^2$. We can easily observe that, the maximum value in the second case is smaller than or equal to that in the first case.

We complete the proof.
\end{proof}

\begin{remark}
When $r$ is small, that is, there is few incentive in the wireless D2D caching network, then each UT only caches the content according to his own interest.
When $r$ is very large, that is, there are enough incentives, then each UT only caches the content according to the interest of the other UT.
\end{remark}

\section{Iterative Gradient Algorithm}
In this section, we present an iterative gradient algorithm (IGA) to obtain the SE for the D2D caching game.

    \begin{algorithm}[h]
    \caption{Iterative Gradient Algorithm}\label{euclid}
    \begin{algorithmic}[1]
   	\State Initialize cache state $\mathbf{X}$
    \Function {UT\_Game}{$\mathbf{X}, r$}

        \For {each UT $i \in \mathcal{N}$ and \ content $m \in \mathcal{M}$}
            \State $grad_i^m = \frac{U_i(x_i+\delta,\ \mathbf{x}_{-i}) - U_i(\mathbf{X})}{\delta}$
        \EndFor
        \State Update $\mathbf{X}$: $x_i^m \leftarrow x_i^m + grad_i^m \times \gamma, \ \forall i\in\mathcal{N},m\in \mathcal{M}$
        \If {$\exists i$ such that $\sum_{m=1}^{M}x_i^ms_m > c_i$}
            \State $\mathbf{x}_i \leftarrow \omega\mathbf{x}_i$ so that cache size constraint is satisfied
        \EndIf
    \State Repeat steps from 3 to 9 until convergence
    \EndFunction
    \State

    \Function {BS\_Opt}{$\mathbf{X}, r$}
    \State BS set $r = 0$
        \State \Call{UT\_Game}{$\mathbf{X}, r$}
        \State Compute BS cost $C(\mathbf{X},r)$
        \State $r \leftarrow r + \Delta_r$
    \State Repeat steps from 15 to 17 until $C$ starts to increase w.r.t $r$
    \EndFunction

        \end{algorithmic}
    \end{algorithm}

The iterative gradient algorithm involves numerous rounds of interactions between the BS and the UTs. The BS starts by setting the unit reward to zero and starts the game. At the followers' side, the goal is to achieve the sub-game Nash equilibrium. We use gradient projection method to reach the equilibrium point, where $\gamma$ is the step size. All UTs simultaneously update their cache state. After each update of cache state $\mathbf{x}_i$, we have to check whether $\mathbf{x}_i$ still belongs to the feasible domain. When all UTs' utilities and cache states converges, the UT sub-game obtains its equilibrium point.

The BS computes his total cost based on the cache state of UTs. The BS then increases the unit reward $r$ by a tiny amount $\Delta_r$ and proceeds the above interactions iteratively. The BS total cost $C$ decreases with regard to $r$ when $r$ is small. The decrease of BS total cost is due to the changing of caching policy of UTs. The UTs are motivated by higher rewards to cache more amount of contents to serve other UTs, and thus reduce the serving  cost of BS. This decreasing trend however, cannot last for long since the benefits BS can get from ``better'' cache states is bounded. After a number of iterations, the rapid increase of reward cost would finally outstrip the decrease of serving cost. The BS has to find the point where his total cost $C$ ceases to decrease, which is the SE of the proposed game.

Note that the UTs' cache state can be initialized using different approaches, such as initializing to zero or caching the favorite contents. The choice of initialization method does not affect the convergence in UT Sub-game. In addition, the values of $\delta$, $\gamma$ and $\Delta_r$ should be sufficiently small to guarantee the convergence of the algorithm.

\section{Performance Evaluation}
In this section, we present the numerical results to evaluate the effectiveness of the proposed Stackelberg game model and the algorithm.
\subsection{Simulation Setup}
For the numerical analysis, we consider a D2D caching network with $N=8$ UTs and $M=20$ contents. Every UT has the same cache size of $c_i=2$, and can communicate with any UT in the network. The size of each content $s_m$ is equal, and is normalized to one. We also set the weight parameters $w_d = 0.5$ and $w_s = 20$.

We set the delay cost of each pair of UTs to be uniformly picked within the range $\left(0,1\right)$ and the delay cost between UT and the BS within the range of $\left(1,6\right)$. The symmetric delay cost matrix $\mathbf{D}$ is given as follows, where the element $\mathbf{D}_{i,j}$ is the delay cost between UT $i$ and UT $j$ for $j < N+1$, and $\mathbf{D}_{i,N+1}$ is the delay cost between UT $i$ and the BS.

\begin{scriptsize}
\begin{equation*}
\mathbf{D} =
\begin{bmatrix}
0 &   &   &   & \cdots & & &  & \\
0.916 & 0 &   &   & \cdots & & &  & \\
0.806 & 0.633 & 0 &   &   &   &   &   &  \\
0.117 & 0.171 & 0.298 & 0 &   &   &   &   &  \\
0.963 & 0.717 & 0.791 & 0.666 & 0 &   &   &   & \\
0.017 & 0.683 & 0.228 & 0.521 & 0.297 & 0 &   &   &  \\
0.090 & 0.180 & 0.828 & 0.798 & 0.499 & 0.534 & 0 &   &  \\
0.904 & 0.097 & 0.973 & 0.402 & 0.863 & 0.896 & 0.879 & 0 &  \\
2.297 & 5.398 & 2.030 & 3.206 & 5.947 & 5.201 & 3.546 & 1.040 & 0
\end{bmatrix}
\end{equation*}
\end{scriptsize}

We use Zipf distribution \cite{zipf} to model the request pattern of a certain UT. The Zipf parameter $\alpha$ determines the skewness of the popularity distribution. To characterize the heterogeneous request pattern among UTs, we randomly permutate the request probability vector and assign it to different UTs.
\subsection{Baselines and Performance Criteria}
We compare the performance of the following caching schemes:
\begin{enumerate}
\item Random Complete Caching (RCC): Each UT $i\in\mathcal{N}$  randomly selects $c_i$ out of $M$ contents to cache.
\item Greedy Caching (GC): Each UT $i\in\mathcal{N}$  caches his top $c_i$ favorite contents.
\item Preference-Aware Caching (PAC): Each UT $i$ allocates $p_i^k$ of his cache size for content $k$.
\item Fair Caching (FC): Each UT caches the same proportion of every content.
\item Stackelberg Caching (SC): The proposed caching scheme using Stackelberg game model.
\end{enumerate}

The performance criteria we consider are the BS's serving cost and UTs' utilities, given by equation (4).

\subsection{Numerical Results}
\subsubsection{Effect of Reward}
We first study the effects of the reward on the BS and the UTs using the proposed iterative gradient algorithm.

The evolution of the BS total cost and serving cost is shown in Figure~\ref{fig:BScosts}. The total cost of BS consists of the reward given to UTs and the cost of serving UTs' requests. For the BS serving cost, we observe that it first quickly decreases as the unit reward increases. After a few iterations, the decrease in serving cost becomes slower and it eventually decreases to a constant. As the reward gets higher, the UTs are more motivated to cache in favor of D2D, thus reducing the serving cost from BS. However the margin gain is gradually diminishing as long as most of the favorite contents of each UT have already been cached within his neighborhood.

As for the BS total cost, the trend is similar to that of serving cost when the reward is low. As unit reward $r$ further increases, the BS total cost increases approximately linearly.  This is due to the diminishing margin gain of BS serving cost and the linear growth of BS reward cost.

The evolution of average UT utility is shown in Figure~\ref{fig:UTevo}.  For completeness, we also depict the minimum and maximum values of UT utility with error bars. At the start of the game, the utilities of all UTs are negative since their utilities are entirely delay cost.  As the unit reward from the BS increases, the utility of each UT increases and eventually becomes positive.
\begin{figure}[t]
\begin{minipage}[t]{.5\linewidth}
\centering
        \includegraphics[width=4.4cm]{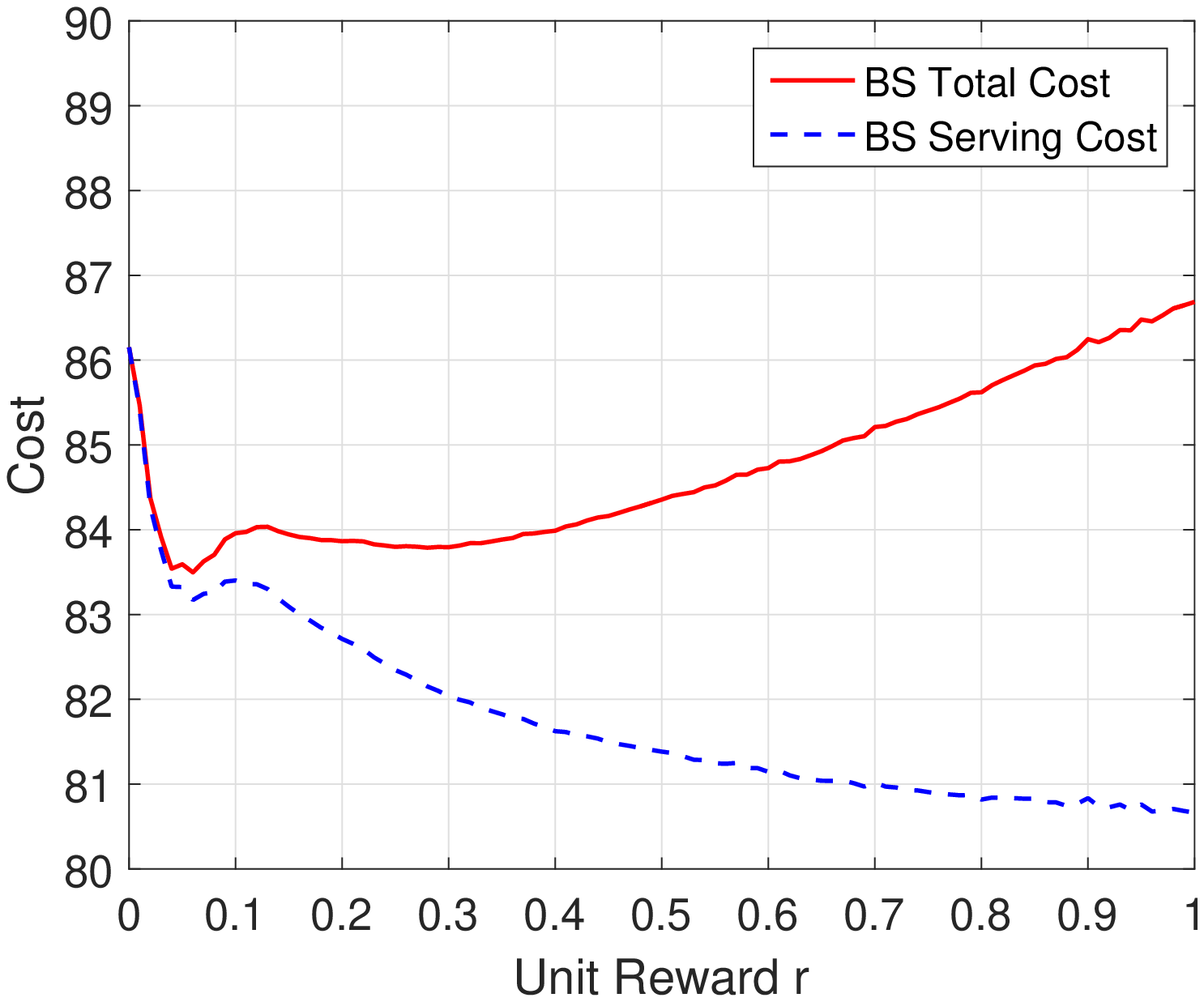}
\subcaption{BS total cost and serving cost.}\label{fig:BScosts}
\end{minipage}%
\begin{minipage}[t]{.5\linewidth}
\centering
        \includegraphics[width=4.4cm]{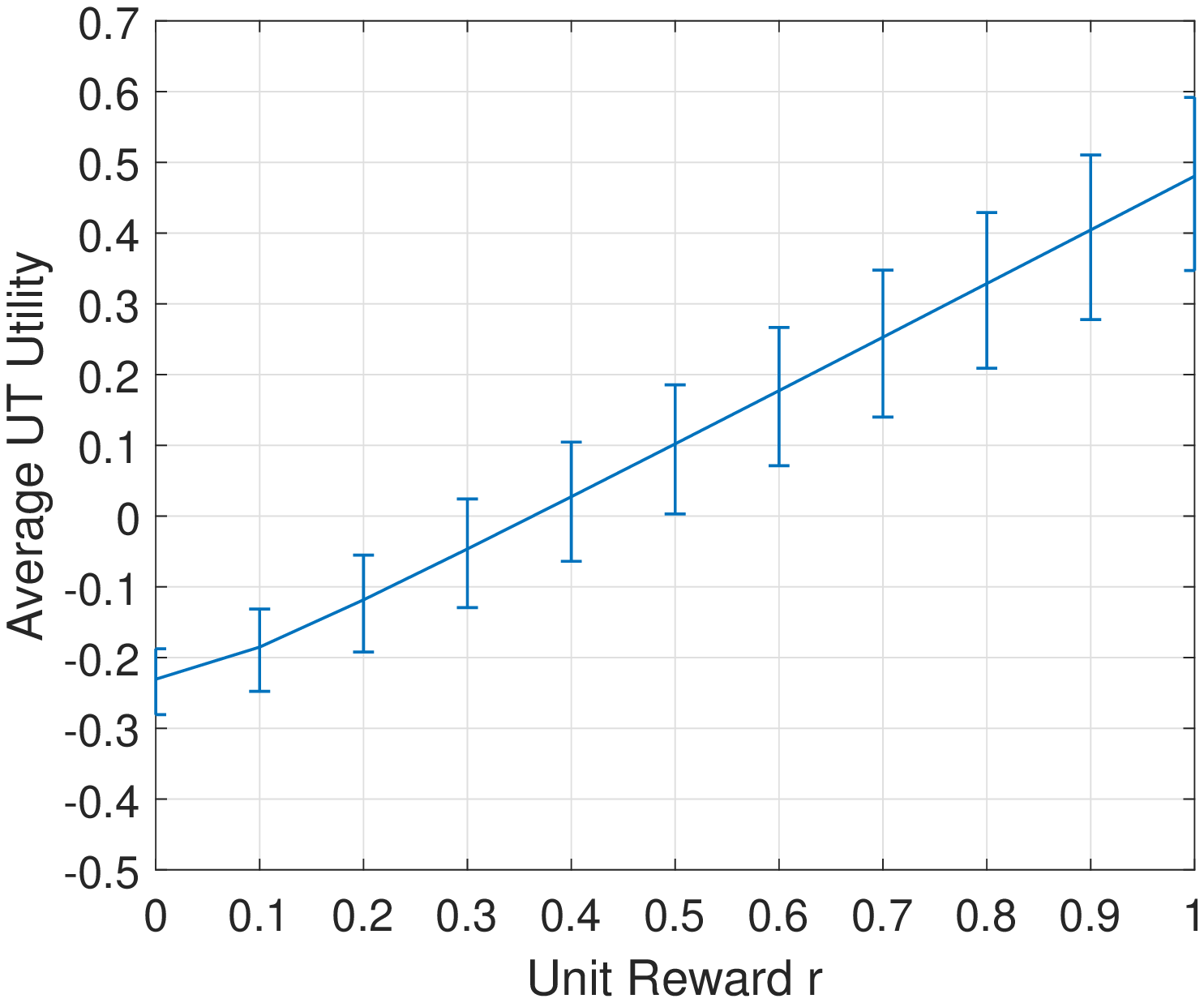}
\subcaption{UT utilities.}\label{fig:UTevo}
\end{minipage}
\caption{Effect of reward on BS cost and UT cost.}\label{fig:uniform}
\end{figure}

\subsubsection{Comparison of Caching Schemes}

In Figure~\ref{fig:BScomp}, we compare the performance of the proposed Stackelberg Caching scheme with the four baseline schemes in Section VI-B.
We see that the BS serving cost under the proposed caching scheme is the lowest among all caching schemes. In particular, the BS serving cost of the proposed SC is lower than that of PAC, GC, RCC, and FC by $12\%$, $37\%$, $68\%$, and $27\%$, respectively.

For convenience of illustration, we define the UT cost as the negative of UT utility and depict the UT costs of the five schemes in Figure~\ref{fig:UTcomp}. We randomly select two UTs and depict their costs as well as the average UT cost for all schemes. We observe that on average, the UT cost under the proposed caching scheme is lower than that of PAC, RCC and FC, and does not possess much advantage over the Greedy Caching. In particular, the average UT cost of SC is lower than that of PAC, RCC, and FC by $9\%$, $60\%$, and $39\%$, respectively.

\begin{figure}[t]
\begin{minipage}[t]{.5\linewidth}
\centering
        \includegraphics[scale=.32]{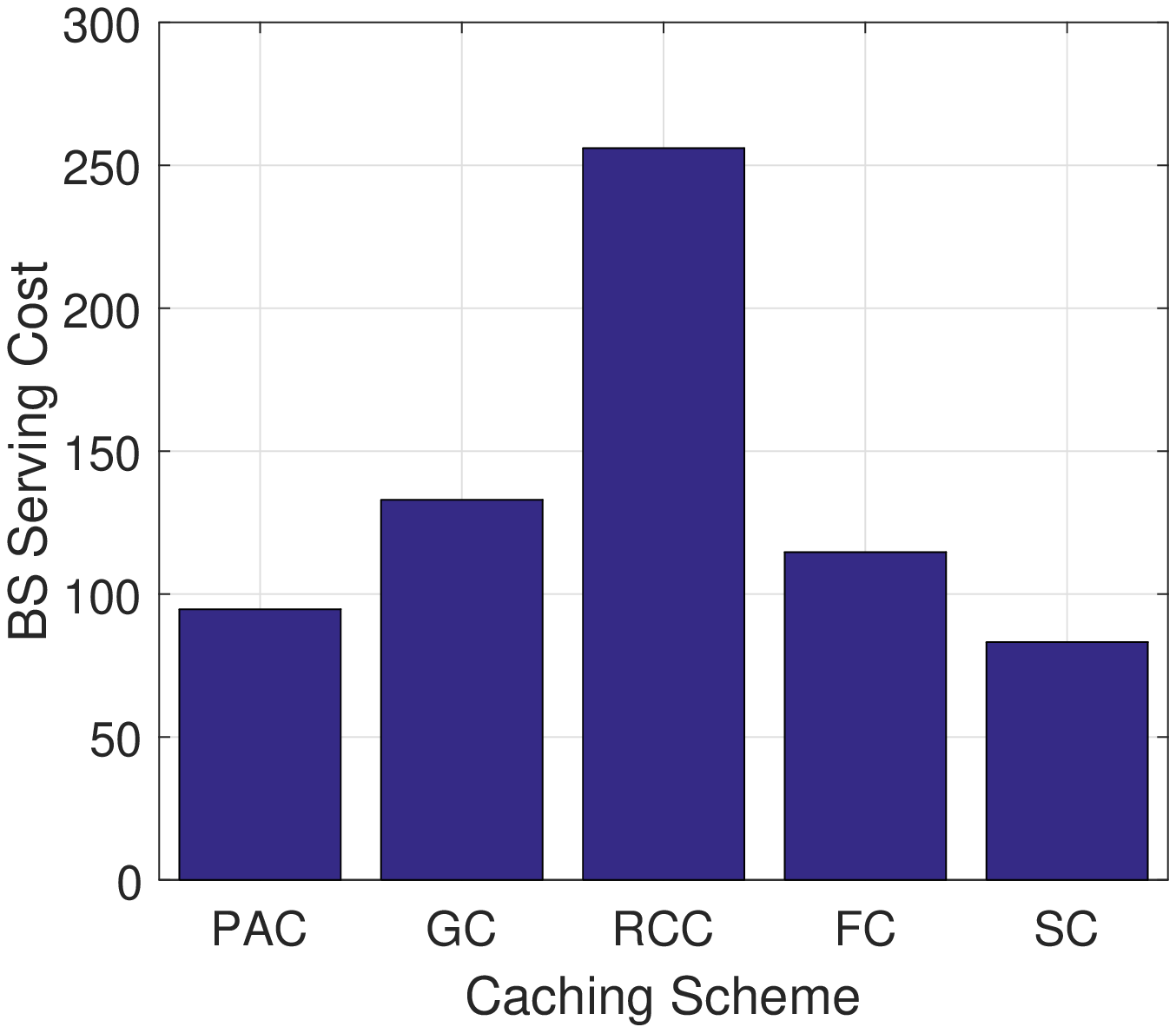}
\subcaption{Comparison on BS serving cost.}\label{fig:BScomp}
\end{minipage}%
\begin{minipage}[t]{.5\linewidth}
\centering
        \includegraphics[scale=.32]{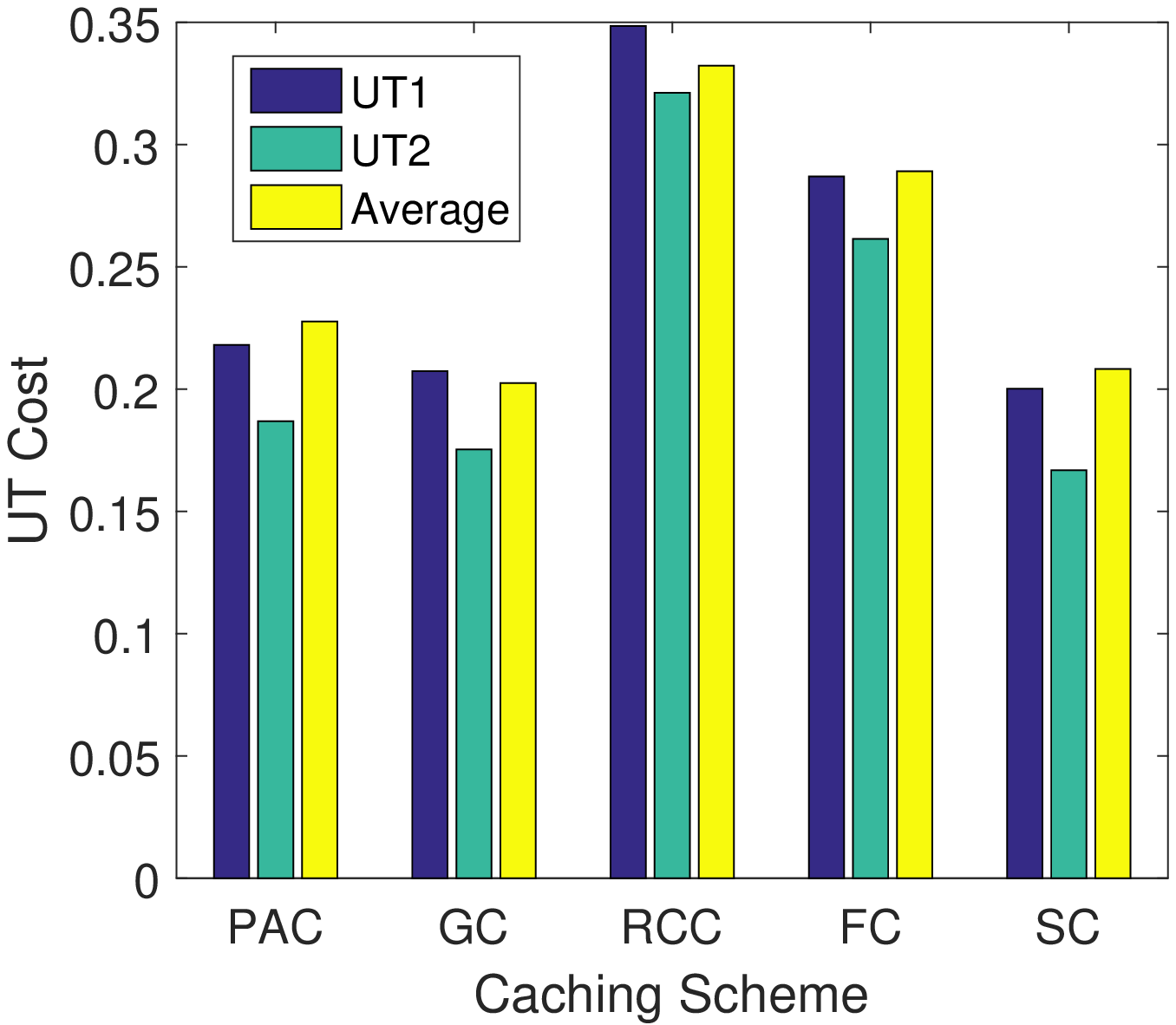}
\subcaption{Comparison on UT cost.}\label{fig:UTcomp}
\end{minipage}
\caption{Comparison of different caching schemes.}\label{fig:Comp}
\end{figure}

%

\subsubsection{Effect of Zipf Parameter}
We illustrate the effect of UTs' preference on the average UT utility of the proposed scheme and the aforementioned four baseline schemes, as illustrated in Figure~\ref{fig:ZipfUT}.
We see that as $\alpha$ increases (from 0.4 to 1.5), the average UT utilities of SC, PAC, GC increases. We can observe that the average UT utility of our proposed scheme is larger than that of any other baseline schemes for all $\alpha$. In addition, the gain of UT utility under Stackelberg Caching over other schemes increases as $\alpha$ increases. This indicates that the proposed caching scheme is more beneficial to UTs if UTs' request pattern is more heterogeneous.
\begin{figure}[htbp]
\centering
\includegraphics[width=7cm]{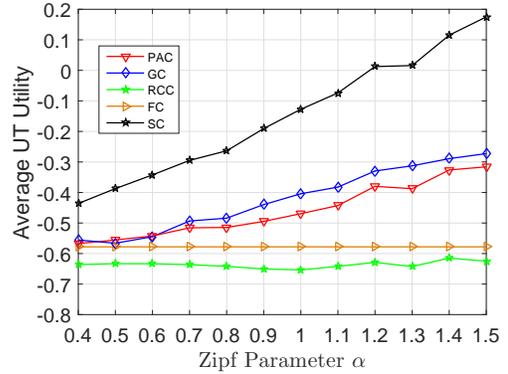}
\caption{Effect of Zipf parameter on average UT utility.}\label{fig:ZipfUT}
\end{figure}


\section{Conclusion}
In this paper, we address the cost minimization problem for the BS by introducing an incentive mechanism to encourage more content sharing among UTs. We formulate the conflict among UTs, as well as the tension between BS and UTs as a Stackelberg game. The BS determines the rewarding policy to minimize his total cost. On the other hand, each UT aims to maximize his utility by choosing his caching policy.  We show the existence of the equilibrium and propose an iterative gradient algorithm (IGA) to obtain the Stackelberg Equilibrium. We also study the impact of incentives on caching strategies analytically for a special case. We compare the performance of the caching scheme at SE with four baseline caching schemes with no incentives. We show that the caching scheme under the proposed incentive mechanism performs better than other baseline schemes in terms of lower BS serving cost and higher average utility for UTs.

\bibliographystyle{IEEEtran}
\bibliography{IEEEabrv,icc16}

\begin{thebibliography}{10}
\providecommand{\url}[1]{#1}
\csname url@samestyle\endcsname
\providecommand{\newblock}{\relax}
\providecommand{\bibinfo}[2]{#2}
\providecommand{\BIBentrySTDinterwordspacing}{\spaceskip=0pt\relax}
\providecommand{\BIBentryALTinterwordstretchfactor}{4}
\providecommand{\BIBentryALTinterwordspacing}{\spaceskip=\fontdimen2\font plus
\BIBentryALTinterwordstretchfactor\fontdimen3\font minus
  \fontdimen4\font\relax}
\providecommand{\BIBforeignlanguage}[2]{{%
\expandafter\ifx\csname l@#1\endcsname\relax
\typeout{** WARNING: IEEEtran.bst: No hyphenation pattern has been}%
\typeout{** loaded for the language `#1'. Using the pattern for}%
\typeout{** the default language instead.}%
\else
\language=\csname l@#1\endcsname
\fi
#2}}
\providecommand{\BIBdecl}{\relax}
\BIBdecl

\bibitem{Cisco}
Cisco, ``Cisco visual networking index: Global mobile data traffic forecast
  update, 2014-2019,'' \emph{White Paper, February}, 2015.

\bibitem{femto}
K.~Shanmugam, N.~Golrezaei, A.~Dimakis, A.~Molisch, and G.~Caire,
  ``Femtocaching: Wireless content delivery through distributed caching
  helpers,'' \emph{{IEEE} Trans. Inf. Theory}, vol.~59, no.~12, Dec 2013.

\bibitem{song2015wireless}
L.~Song, D.~Niyato, Z.~Han, and E.~Hossain, \emph{Wireless Device-to-Device
  Communications and Networks}.\hskip 1em plus 0.5em minus 0.4em\relax
  Cambridge University Press, 2015.

\bibitem{molisch2014caching}
A.~F. Molisch, G.~Caire, D.~Ott, J.~R. Foerster, D.~Bethanabhotla, and M.~Ji,
  ``Caching eliminates the wireless bottleneck in video aware wireless
  networks,'' \emph{Advances in Electrical Engineering}, 2014.

\bibitem{myerson1991game}
R.~B. Myerson, ``Game theory: analysis of conflict,'' \emph{Harvard
  University}, 1991.

\bibitem{poularakis2014framework}
K.~Poularakis, G.~Iosifidis, and L.~Tassiulas, ``A framework for mobile data
  offloading to leased cache-endowed small cell networks,'' in \emph{Proc. IEEE
  MASS}, 2014, pp. 327--335.

\bibitem{wu2012game}
W.~Wu, J.~Lui, and R.~T. Ma, ``A game theoretic analysis on incentive
  mechanisms for wireless ad hoc vod systems,'' in \emph{Proc. IEEE WiOpt},
  2012, pp. 177--184.

\bibitem{alotaibi2015game}
F.~Alotaibi, S.~Hosny, H.~El~Gamal, and A.~Eryilmaz, ``A game theoretic
  approach to content trading in proactive wireless networks,'' in \emph{Proc.
  IEEE ISIT}, June 2015, pp. 2216--2220.

\bibitem{shokrollahi2006raptor}
A.~Shokrollahi, ``Raptor codes,'' \emph{{IEEE} Trans. Inf. Theory}, vol.~52,
  no.~6, pp. 2551--2567, 2006.

\bibitem{rosen1965existence}
J.~B. Rosen, ``Existence and uniqueness of equilibrium points for concave
  n-person games,'' \emph{Econometrica: Journal of the Econometric Society},
  pp. 520--534, 1965.

\bibitem{goebel1990topics}
K.~Goebel and W.~A. Kirk, \emph{Topics in metric fixed point theory}.\hskip 1em
  plus 0.5em minus 0.4em\relax Cambridge University Press, 1990, vol.~28.

\bibitem{zipf}
L.~Breslau, P.~Cao, L.~Fan, G.~Phillips, and S.~Shenker, ``Web caching and
  {Zipf-like} distributions: evidence and implications,'' in \emph{Proc. IEEE
  INFOCOM}, March 1999.

\end{thebibliography}

\end{document}